\def\year{2022}\relax
\documentclass[letterpaper]{article} %
\usepackage{aaai22}  %
\usepackage{times}  %
\usepackage{helvet}  %
\usepackage{courier}  %
\usepackage[hyphens]{url}  %
\usepackage{graphicx} %
\usepackage{natbib}  %
\usepackage{caption} %
\DeclareCaptionStyle{ruled}{labelfont=normalfont,labelsep=colon,strut=off} %
\usepackage{algorithm}
\usepackage{algorithmic}

\usepackage{tikz}
\usepackage{booktabs}

\usepackage[none]{hyphenat}
\usepackage{newfloat}
\usepackage{listings}
\floatstyle{ruled}
\newfloat{listing}{tb}{lst}{}
\floatname{listing}{Listing}
\usepackage{xcolor}
\usepackage{enumerate}

\title{Protein Folding Neural Networks Are Not Robust}
\author{
    Sumit Kumar Jha\textsuperscript{\rm 1},
    Arvind Ramanathan\textsuperscript{\rm 2},
    Rickard Ewetz\textsuperscript{\rm 3},
    Alvaro Velasquez\textsuperscript{\rm 4},
    Susmit Jha\textsuperscript{\rm 5}\\
}
\affiliations{
    \textsuperscript{\rm 1}
    Computer Science Department,
    University of Texas at San Antonio, TX 78249 \\
    \textsuperscript{\rm 2} Data Science and Learning, Argonne National Laboratory, Lemont, IL, 60439 \\
    \textsuperscript{\rm 3} Electrical and Computer Engineering Department, University of Central Florida, Orlando, FL 32816 \\
    \textsuperscript{\rm 4} Information Directorate, Air Force Research Laboratory, Rome, NY 13441 \\
    \textsuperscript{\rm 5} Computer Science Laboratory, SRI International, Menlo Park, CA, 94709 \\
    sumit.jha@utsa.edu, ramanathana@anl.gov, rickard.ewetz@ucf.edu, alvaro.velasquez.1@us.af.mil, susmit.jha@sri.com
}

\usepackage{bibentry}

\newtheorem{theorem}{Theorem}
\newtheorem{proof}{Proof}

\newcommand{\noop}[1]{}

\begin{document}

\maketitle
\begin{abstract}
Deep neural networks such as  AlphaFold and RoseTTAFold predict remarkably accurate structures of proteins compared to other algorithmic approaches. It is known that biologically small perturbations in the protein sequence do not lead to drastic changes in the protein structure. In this paper, we demonstrate that RoseTTAFold does not exhibit such a robustness despite its high accuracy, and biologically small perturbations for some input sequences result in radically different predicted protein structures. 
This raises the challenge of detecting when these predicted protein structures cannot be trusted. 
We define the robustness measure for the predicted structure of a protein sequence to be the inverse of the 
root-mean-square distance (RMSD) in the predicted structure and the structure of its adversarially 
perturbed  sequence. 
We use adversarial attack methods to create adversarial protein sequences, and 
show that the RMSD in the predicted protein structure ranges from  0.119\r{A} to 34.162\r{A} when the adversarial perturbations are bounded by 20 units in the BLOSUM62 distance.
This demonstrates very high variance in the robustness measure of the predicted structures. 
We show that the magnitude of the correlation (0.917) between our robustness measure 
and the RMSD between the predicted structure and the ground truth is high, that is, the predictions with low robustness measure cannot be trusted. This is the first paper demonstrating the susceptibility of RoseTTAFold to adversarial attacks.

\end{abstract}

\section{Introduction}
Proteins form the building blocks of life as they enable a variety of vital functions essential to life and reproduction. Naturally occurring proteins are bio-polymers composed of 20 amino acids and this primary sequence of amino acids is well known for many proteins, thanks to high throughput sequencing techniques. However, in order to understand the functions of different protein molecules and complexes, it is essential to comprehend their three-dimensional (3D) structure. 
Until recently, one of the grand challenges in structural biology has been the accurate determination of the 3D structure of the protein from its primary sequence. Accurate predictive protein folding promises to have a profound impact on the design of therapeutics for diseases and drug discovery \cite{chan2019advancing}. 
\begin{figure}[t]
\begin{tabular}{cc}
 \hspace{-0.7cm} \includegraphics[width=0.58 \columnwidth]{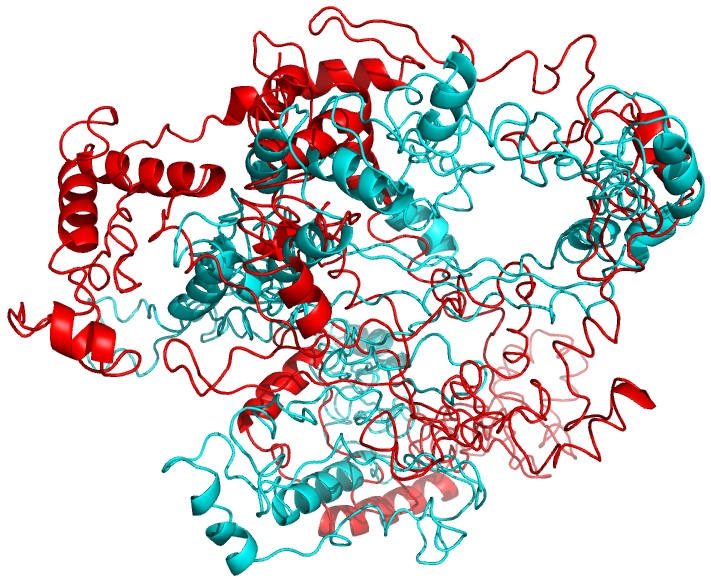}    & \hspace{-0.5cm} \includegraphics[width=0.49 \columnwidth]{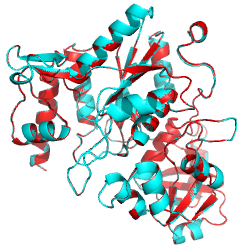} \\
   (i) RMSD=34.162\r{A}  & (ii) RMSD=0.119\r{A}
\end{tabular} 
\caption{The predicted structure for the original (shown in blue) and the adversarial (shown in red) protein sequences for the SFP1 interferon stimulator (left) and 2QJF\_1 human bifunctional 3'-phosphoadenosine 5'-phosphosulfate synthetase 1 (right) as predicted by RoseTTAFold and aligned using PyMOL. RMSD is the Root Mean Square Distance, typically measured in \r{A} $ (10^{-10}$m). The predicted structure on the right is robust to biologically small perturbations, while the one on the left is too sensitive even to biologically small changes. }
\label{Fig:intro}
\end{figure}

AlphaFold achieved unparalleled success in predicting protein structures using neural networks and remains first at the CASP2020 competition. RoseTTAFold is the first open source deep learning model for predicting protein structure with high accuracy. 
In this paper, we use adversarial attack methods to investigate the ability of these models to generalize to new sequences and  propose a robustness metric that identifies input protein sequences on which the output of these protein folding neural networks (PFNN) cannot be trusted. We make the following contributions in this paper:

\begin{enumerate}[(i)]
    \item We demonstrate the susceptibility of RoseTTAFold to adversarial attacks~\cite{goodfellow2018making} by generating several examples where protein sequences which vary only in five residues result in very different 3D protein structures. We use sequence alignment scores~\cite{henikoff1992amino} such as those derived from Block Substitution Matrices (BLOSUM62) to identify a space of similar protein sequences used in constructing adversarial perturbations. %
    We define a new robustness measure for the predicted structure of a protein sequence to be the inverse of the root-mean-square distance (RMSD) in the predicted structure and the structure of its adversarially perturbed  sequence. 
    \item We propose the use of different distance measures based on co-ordinate geometry and distance geometry~\cite{blumenthal1970theory}, including a fast linear-time provably approximate variant of the distance geometry metric. The latter is invariant under rigid-body motion and, hence, does not suffer from the deficiencies of the co-ordinate geometry representation during adversarial attacks. Our distance geometry based approach produces adversarial examples that have up to 37 times larger RMSD than those produced using co-ordinate geometry based representation. Further, our approximate distance geometry approach is provably approximate and produces no less than 0.69 times the exact RMSD while achieving 14X speedups in our experiments on a 80GB A100 GPU system.
    \item Our  experiments show that different input protein sequences have very different adversarial robustness.
    RMSD in the  protein structure predicted by RoseTTAFold~\cite{baek2021accurate}  ranges from  0.119\r{A} to 34.162\r{A} when the adversarial perturbations are bounded by 20 units in the BLOSUM62 distance.
     Hence, our proposed approach can be used to identify protein sequences on which the predicted 3D structure  cannot be trusted.
     We show that the magnitude of the correlation (0.917) between our robustness measure and the RMSD between the predicted structure and the ground truth is high. This compares favorably with DeepAccNet~\cite{hiranuma2021improved} (correlation magnitude: 0.869) used for estimating accuracy of predicted protein structures. 

\end{enumerate}

\section{Summary and Related Work}

\subsection{Motivation}

PFNNs~\cite{jumper2021highly,baek2021accurate} should be expected to obey the natural observation that biologically small changes in the sequence of a protein usually do not lead to drastic changes in the protein structure. Almost four decades ago, it was noted that two structures with  50\% sequence identity align within approximately 1\r{A} RMSD from each other~\cite{chothia1986relation}. Two proteins with even 40\% sequence identity and at least 35 aligned residues align within approximately 2.5\r{A}~\cite{sander1991database}. The phenomenon of sequence-similar proteins producing similar structures have also been observed in larger studies~\cite{rost1999twilight}. As with almost any rule in biology, a small number of counterexamples to the similar sequence leading to similar structure dogma do exist, where by even small perturbations can potentially alter the entire fold of a protein. However, such exceptions are not frequent and often lead to exciting investigations~\cite{Cordes_2000,tuinstra2008interconversion}.

\subsection{Robustness Metric using Adversarial Attacks}
The similar-sequence implies similar-structure paradigm dictates that PFNNs should make robust predictions. %
Given a protein sequence $S$ with a three-dimensional structure $\mathcal{A}(S)$, we define a notion of biologically similar sequence structures $\mathcal{S}$ using Block Substitution Matrices (BLOSUM)~\cite{henikoff1992amino}, and then employ adversarial attacks~\cite{goodfellow2018making} on PFNNs within this space of similar sequences to identify a sequence $S_{adv} \in \mathcal{S}$ that produces a maximally different three-dimensional structure $\mathcal{A}(S_{adv})$. We then compute a distance $D_{str}(\mathcal{A}(S), \mathcal{A}(S_{adv}))$ between the structures for the original input and the adversarial structure, and use the inverse of this distance as the robustness  measure. If the distance is small, the response of the PFNN is robust; a large distance indicates that the predicted structure is not robust.  

\subsection{BLOSUM Similarity Measures}
We identify the space of biologically similar sequences $\mathcal{S}$ for a given protein sequence $S$. We expect the predicted structures for the similar sequences to be similar. If there is a large RMSD between
the predicted structure $\mathcal{A}(S)$ and  the structure $\mathcal{A}(S_{adv})$ of the adversarial sequence 
 $S_{adv} \in \mathcal{S}$, 
it would reflect a lack of robustness in the prediction of the network. We adopt a sequence similarity measure that counts replacement frequencies in conserved blocks across different proteins; the BLOSUM matrix $[B_{ij}]$ is a matrix of integers, where each entry denotes the similarity between a residue of type $b_i$ and type $b_j$.

Given two sequences $S=s_1 s_2 \dots s_n$ and $S'=s'_1 s'_2 \dots s'_n$ where $T(s_i)$ denotes the type of the residue $s_i$, the BLOSUM distance between the two sequences is given by $D_{seq}(S,S') = \sum_{1 \leq i \leq n} \left( B_{{T(s_i)}{T(s_i)}} - B_{{T(s_i)}{T(s'_i)}} \right)$. If $s_i=s'_i$ i.e. the two sequences have the same residue, the contribution of this residue is 0; otherwise, the contribution of the change in residue to the BLOSUM distance is the positive integer $B_{{T(s_i)}{T(s_i)}}$ less the new contribution $B_{{T(s_i)}{T(s'_i)}}$.

\subsection{Distance Geometric Representation}
Adversarial attacks on predicted structures need a distance measure $D$ between two 3D structural predictions $\mathcal{A}(S)$ of sequence $S$ and $\mathcal{A}(S')$ of sequence $S'$. The most straightforward candidate is the distance $D_{co\-ord}$ between the three-dimensional co-ordinates of the predicted structures. However, this leads to scenarios where the adversarial attack causes significant translation and rotation of the structures, and does not necessarily change the structure of the protein that is invariant to rigid-body motion. See Fig.~\ref{Fig:geometric_init} for an illustration of an adversarial attack using co-ordinate geometry representations.
\begin{figure}
    \centering
    \includegraphics[width=4cm]{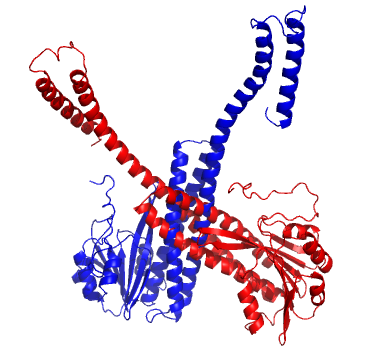}
    \caption{Adversarial attacks using co-ordinate representation may produce structures that are rotations, translations or mirror images of the original structure without being substantially different.}
    \label{Fig:geometric_init}
\end{figure}
Hence, we propose using distance geometry during adversarial attacks to describe a protein structure $S=s_1 s_2 \dots s_n$ by the (element-wise square of the) Euclidean distance matrix $[d(s_i,s_j)]_{n \times n}$ between all pairs $s_i$ and $s_j$ of its residues. Given two such distance geometry representations $[d(s_i,s_j)]_{n \times n}$ and $[d(s'_i,s'_j)]_{n \times n}$ of two sequences $S$ and $S'$, the distance-geometric measure between the two structures is defined as $D_{distgeo} = \sum_{ 1 \leq i,j \leq n} \left( d(s_i,s_j) - d(s'_i,s'_j) \right)^2$. Theorem~\ref{theorem1} shows that this distance measure collapses to 0 only when the two structures can be aligned perfectly; it is invariant to rigid-body motions.

\subsection{Approximate Distance}
Our computational experiments (see Table~\ref{table:geometric}) show that the distance geometry based representations requires more than a day of computation time on a 80GB A100 GPU. These results are consistent with the fact that the distance geometry approach needs $\mathcal{O}(n^2)$ computations for a protein sequence of length $n$. 
Instead of computing the difference between sum of distance between all pairs of two structures, we seek to compute the distance $D_{max} = \left( d(s_a, s_b) - d(s'_a, s'_b) \right)^2$ between the farthest pair of residues $s_a$ and $s_b$ in the sequence $S$ and the corresponding residues $s'_a$ and $s'_b$ in $S'$. Theorem~\ref{theorem2} shows that it is possible to approximate $D_{max}$ using a $\mathcal{O}(n)$ algorithm. Table~\ref{table:approx1} and~\ref{table:approx2} show that this approximation approach is up to 14X faster and still produces adversarial examples that are no less than 0.69 times the exact RMSD.

\section{Approach}
\label{sec:approach}

Our approach to evaluating the robustness of PFNNs is based on two main ideas: (i) the existence of adversarial examples in PFNNs that produce adversarial structures possibly very different from the original structure, and (ii) the use of BLOSUM matrices for identifying a neighborhood of a given sequence that are biologically similar and hence expected to have similar 3D structures. We employ the inverse of the RMSD between the structure of an original protein sequence and the structure of the adversarial sequence as a measure of robustness of a protein folding network on the given input sequence.
\subsection{Adversarial Attacks on PFNN}

Small carefully crafted changes in a few pixels of input images cause well-trained neural networks with otherwise high accuracy to consistently produce incorrect responses in domains such as computer vision~\cite{croce2020robustbench,andriushchenko2020square,bai2020improving,croce2021mind}. 
Given a neural network $\mathcal{A}$ mapping a sequence $S$ of residues to a three-dimensional geometry $\mathcal{A}(S)$ describing the structure of the protein, we seek to obtain a sequence $S'$ such that the sequence similarity measure $D_{seq}(S, S')$ between $S$ and $S'$ is small (that is, it is below a threshold $T$ set to $20$ in our experiments), and the structural distance measure $D_{str}(\mathcal{A}(S), \mathcal{A}(S'))$ is maximized. This can be achieved by solving the following optimization problem:

$$\max_{S'} (D_{str}\left(\mathcal{A}(S), \mathcal{A}(S')\right)  \; s.t. \; D_{seq}(S,S') \leq T$$
\subsection{Structural Distance Measure}
Given a sequence of $n$ residues $S=s_1 s_2 \dots s_n$, its three dimensional structure $\mathcal{A}(S)$ is an ordered $n$-tuple of three-dimensional co-ordinates $(x_1, y_1, z_1), \dots (x_n,y_n,z_n)$. Similarly, the structure of a sequence $S'=s'_1 s'_2 \dots s'_n$ is given by $\mathcal{A}(S') = (x'_1, y'_1, z'_1), \dots (x'_n,y'_n,z'_n)$. Our goal is to develop a structural distance measure that captures the variations in the two structures $\mathcal{A}(S)$ and $\mathcal{A}(S')$ and is invariant to rigid-body motion. Further, we seek to ensure that the structural distance measure can be approximated efficiently in linear time.

\subsubsection{Inadequacy of Co-ordinate Representations} A straightforward solution to define the structural similarity measure would be to use the Euclidean metric over a co-ordinate geometric representation: $D_{coord}(\mathcal{A}(S),\mathcal{A}(S')) =  \sum_{1 \leq i \leq n} \left( (x_i-x'_i)^2+(y_i-y'_i)^2+(z_i-z'_i)^2 \right)$. However, simple translation of the structural co-ordinates of one structure or their relative rigid body rotation causes the Euclidean metric in the co-ordinate geometric representation to increase without really changing the three-dimensional structure of the proteins. See Figures~\ref{Fig:geometric_init} and~\ref{Fig:geometric}
for illustrations showing the inadequacy of co-ordinate geometric representations. 
\subsubsection{Distance Geometry Representation}
In order to compute the distance between two structures, we use the distance geometry representation  and denote each structure in terms of the pairwise distance $d(s_i,s_j)=   \left( (x_i-x_j)^2 + (y_i-y_j)^2 + (z_i-z_j)^2 \right)$ between all pairs of residues in the three-dimensional structure. Hence, we denote the predicted structure $\mathcal{A}(S)$ of the protein using a two-dimensional $n \times n$ matrix of pairwise distance measures $[d(s_i,s_j)]_{n \times n}$. This representation is independent of translation and rigid-body rotation, and adversarial attacks have no incentives for rotating or translating protein structures.

Given the Euclidean distance matrices $[d(s_i,s_j)]_{n \times n}$ and $[d(s'_i,s'_j)]_{n \times n}$ of two sequences $S$ and $S'$ with structures $\mathcal{A}(S)$ and $\mathcal{A}(S')$, we use the Frobenius norm of the difference of the two matrices  $D_{distgeo} = \sum_{1 \leq i,j \leq n} (d(s_i,s_j) - d(s'_i,s'_j))^2 $. The new distance measure $D_{distgeo}$ is invariant to rigid-body motion, including rotations and translations as shown in Theorem~\ref{theorem1} (proof in appendix).

\begin{theorem}[Uniqueness]
\label{theorem1}
Two protein structures are identical up to rigid body transformations if and only if $D_{distgeo}(\mathcal{A}(S),\mathcal{A}(S'))=0$
\end{theorem}
\begin{proof}
Rotations, inversions and translations do not change the pairwise distance between points; so, the ``only if" part of the result is straightforward.

\noindent We prove the other direction by induction on the length $n$ of the proteins. For the base case, we consider $n$=2.
If $D_{distgeo}=0$, $d(s_1,s_2) = d(s'_2,s'_2)$  Now, rewriting in terms of Euclidean co-ordinates, $(x_1-x_2)^2+(y_1-y_2)^2+(z_1-z_2)^2=(x'_1-x'_2)^2+(y'_1-y'_2)^2+(z'_1-z'_2)^2$. We translate the co-ordinate system for $S$ and $S'$ such that its $s_1$ and $s'_1$ are aligned at the origin; then $x_1=x'_1=0, y_1=y'_1=0, z_1=z'_1=0$. Thus, the following follows from $D_{distgeo}=0$: $x_2^2+y_2^2+z_2^2={x'_2}^2+{y'_2}^2+{z'_2}^2 = h^2 \textrm{(say)}$ in the new co-ordinate system. We rotate the co-ordinate system such that $x_2=x'_2=0$ and $y_2=y'_2=0$. Then, $z_2=z'_2=h$ in the new co-ordinates. Thus, the two structures are identical up to translation and rotation.

Our inductive hypothesis assumes that two protein structures of length $n-1$ or less are identical up to rigid body transformations if $D_{distgeo}(\mathcal{A}(S),\mathcal{A}(S'))=0$. 

Inductive Step: We consider two protein structures $S$ and $S'$ of length $n$ such that $D_{distgeo}=0$ i.e. $d(s_i,s_j) = d(s'_i,s'_j)$ for all residues $i$ and $j$ as $D_{distgeo}$ is the sum of squares. Using the inductive hypothesis, we translate the co-ordinate system for $S$(2:n) and $S'$(2:n) such that the two $n-1$ residue fragments are aligned and then translate until $s_1$ is at the origin. Then, rewriting in terms of Euclidean co-ordinates, $(x_i-x_j)^2+(y_i-y_j)^2+(z_i-z_j)^2=(x'_i-x'_j)^2+(y'_i-y'_j)^2+(z'_i-z'_j)^2$ for all residues $i$ and $j$. In particular, $(x_1-x_j)^2+(y_1-y_j)^2+(z_1-z_j)^2=(x'_1-x'_j)^2+(y'_1-y'_j)^2+(z'_1-z'_j)^2$ for all values of $j$. Since the other $n-1$ residues are aligned and $s_1$ is at the origin, this can be simplified to $x_j^2+y_j^2+z_j^2=(x'_1-x_j)^2+(y'_1-y_j)^2+(z'_1-z_j)^2$. Clearly, $x'_1=0, y'_1=0, z'_1=0$ is a solution to this system of equations. Hence, $s'_1$ can also be placed at the origin and the two structures $S$ and $S'$ are identical.
\end{proof}

While the new distance geometric measures are invariant to rigid-body dynamics, they are computationally expensive. Instead of the $\mathcal{O}(n)$ computation required to compute structural distance measures between two  protein structures represented using co-ordinate geometry, the computation of the distance measure using distance geometry $D_{distgeo}$ needs $\mathcal{O}(n^2)$ pairwise computations. As shown in Table~\ref{table:geometric}, this leads to a significant computational challenge and our adversarial robustness computation take more than a day on a 80GB A100 GPU.
\subsubsection{Approximation Strategy}
Another quantity that is invariant to translations and rotations but still useful in measuring the structural changes is the maximum distance between two residues in a protein structure.  Let $s_a$ and $s_b$ be the farthest residues in the protein $S$.  We suggest replacing the sum computation in $D_{distgeo}$ by the  computation $D_{max}(S,S') = (d(s_a,s_b) - d(s'_a,s'_b))^2$. Here, we assume that 
$d(s'_a,s'_b)=\alpha d(s_s,s_b)$ for the farthest pair $a, b$ and some distortion $\alpha$ depending on the sequence. Instead of computing the sum of all changes, this measure computes the impact of the adversarial attack on the farthest residues. 

\begin{theorem}[Strategy for $4(1-\alpha)^2$-Approximation]
\label{theorem2}
Consider a residue $s_k$ in the protein and determine the farthest residue $s_l$ from this residue in the protein. Then, let $s_m$ be the farthest residue from $s_l$. Given a distortion factor $\alpha$, $d^2(s_m,s_l)$ is a $4(1-\alpha)^2$-approximation for $D_{max}$.
\end{theorem}
\begin{proof}
\label{proof:triangle}
We know that $s_a$ and $s_b$ are the farthest residues in a protein $S$ and further $D_{max}(S, S') = (1-\alpha)^2 d^2(s_a,s_b)$ by definition of $D_{max}$ and $\alpha$.

Now, consider points $s_a, s_b, s_m$ and $s_l$. By choice of $s_m$, we know that $d(s_l,s_m)$ is at least as large as $d(s_l,s_a)$ and $d(s_l,s_b)$ i.e. $d(s_l,s_m) \geq d(s_l,s_a)$ and $d(s_l,s_m) \geq d(s_l,s_b)$. 

\noindent Then, applying triangle inequality to $\Delta s_a s_b s_l$, we get $d(s_a, s_b) \leq d(s_l,s_a) + d(s_l, s_b)$. 

\noindent Putting the above two steps together, we get $d(s_a,s_b) \leq 2 d(s_l,s_m)$. Thus, $D_{max}(S) \leq 4 (1-\alpha)^2  d^2(s_l,s_m)$.
\end{proof}

The approximate distance measure using distance geometry $D_{approx}(S,S') = d^2(s_m,s_l)$ is a $4(1-\alpha)^2$-approximation to the square of the change in the maximum distance between pairs of points $D_{max}$ as shown in Theorem~\ref{theorem2} (proof in appendix). Thus, the distance measure $D_{approx}(S,S')$ provides a $\mathcal{O}(n)$ approach for computing change in protein structures during adversarial attacks that are invariant to translations and rotations.
\subsection{Sequence  Similarity Measures}

Given two sequences of $n$ residues $S=s_1 s_2 \dots s_n$ and $S'=s'_1 s'_2 \dots s'_n$, a natural question is to compute the sequence similarity $D_{seq}$ between these proteins. A naive approach would be to count only the number of residues that are different; however, an analysis of naturally occurring proteins shows that all changes in residues do not have the same impact on protein structures. Changes to one type of residue are more likely to cause structural variations than changes to another type of residue.

Early work in bioinformatics focused on properties of amino acids and reliance on genetic codes. However, more modern methods have relied on the creation of amino-acid scoring matrices that are derived from empirical observations of frequencies of amino acid replacements in homologous sequences~\cite{dayhoff197822,jones1992rapid}.

The original scoring matrix, called the PAM250 matrix, was based on empirical analysis of $1,572$ mutations observed in $71$ different families of closely-related proteins which are 85\% or more identical after they have been aligned. The PAM1 model-based scoring matrix was obtained by normalizing the frequency of mutations to achieve a 99\% identity between homologous proteins. These results were then extrapolated to create scoring matrices PAM10, PAM30, PAM70 and PAM120 with 90\%, 75\%, 55\%, and 37\% identity between homologous proteins.

\begin{figure*}
    \centering
\begin{tabular}{ccc}
 \includegraphics[width=0.5 \columnwidth]{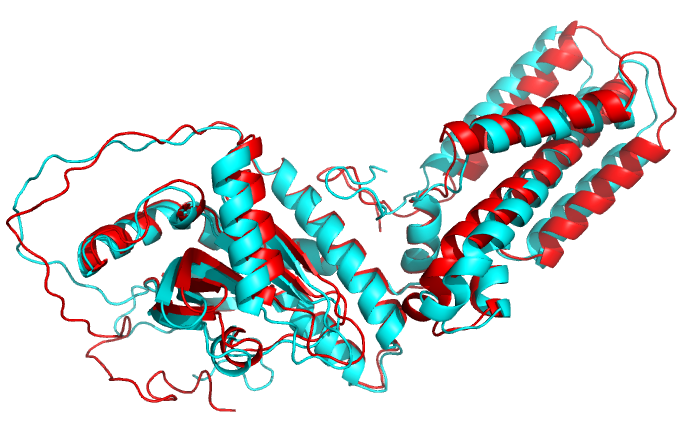} 
 & \includegraphics[width=0.5 \columnwidth]{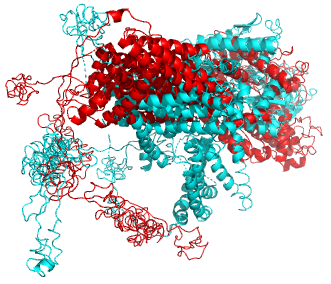} 
 &  \includegraphics[width=0.5 \columnwidth]{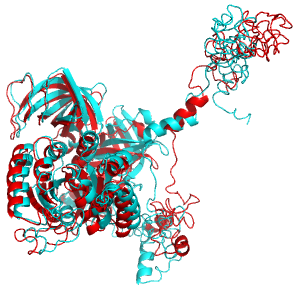} 
\\
 (a) 6NT5\_1
 & (b) 60EU\_1
 & (c) NP\_010457.3\\
 \includegraphics[width=0.5 \columnwidth]{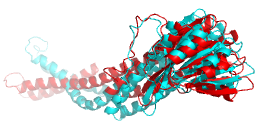} 
 & \includegraphics[width=0.5 \columnwidth]{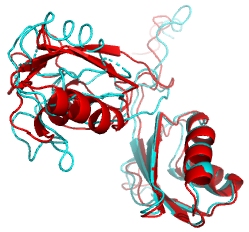} 
 & \includegraphics[width=0.5 \columnwidth]{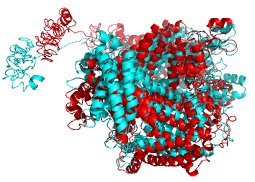} 
\\
 (d) 2BBJ\_1
 & (e) 2ROP\_1
 & (f) 6NT3\_1\\
\end{tabular}

\caption{The predicted structure of protein sequences (blue) and the structures of their adversarial perturbations (red) produced using RoseTTAFold and aligned using PyMOL. }
    \label{Fig:case_studies}
\end{figure*}

Another interesting approach~\cite{henikoff1992amino} to understanding protein similarity is the direct counting of replacement frequencies using the so-called Block Substitution Matrices (BLOSUM). Instead of relying solely on sequences of homologous proteins that are relatively harder to find, the new approach focuses on identifying conserved blocks or conserved sub-sequences in a larger variety of proteins potentially unrelated by evolutionary pathways. The BLOSUM approach then counts the frequency of replacements within these conserved sub-sequences. BLOSUM 62, BLOSUM 80 and BLOSUM 90 denote block substitution matrices that are obtained from blocks or subsequences with at least 62\%, 80\% and 90\% similarity, respectively. The BLOSUM matrix $[B_ij]$ is a matrix of integers where each entry denotes the similarity between residue of type $b_i$ and type $b_j$

Given two sequences $S=s_1 s_2 \dots s_n$ and $S'=s'_1 s'_2 \dots s'_n$ where $T(s_i)$ denotes the type of the residue $s_i$, the BLOSUM distance between the two sequences is given by $D_{seq}(S,S') = \sum_{1 \leq i \leq n} \left( B_{{T(s_i)}{T(s_i)}} - B_{{T(s_i)}{T(s'_i)}} \right)$. If $s_i=s'_i$, i.e. the two sequences have the same residue, the contribution of this residue is 0; otherwise, the contribution of the change in residue to the BLOSUM distance is the positive integer $B_{{T(s_i)}{T(s_i)}}$ less the new contribution $B_{{T(s_i)}{T(s'_i)}}$. In our experiments, we have used a sequence similarity measure derived from the popularly used BLOSUM62 matrix.

\section{Experimental Results}

Our computational experiments employ the RoseTTAFold protein structure prediction neural network and investigate the following questions:
\begin{enumerate}
    \item Can we show that adversarial attacks on protein folding networks demonstrate that the predicted 3D structures are not robust for some protein sequences?
    If so, does our robustness measure (where higher robustness corresponds to low RMSDs) correlate with the accuracy of the predicted 3D protein structures?
    \item Do larger BLOSUM62 distances (corresponding to low biological similarity) between adversarial and original sequences lead to higher RMSDs in predicted adversarial structures?
    \item Can three-dimensional co-ordinate representations correctly estimate RMSDs of adversarial examples?
    \item Is our approximation strategy for maximum distance between pairs of residues justified in terms of computational time and the accuracy of RMSDs?
\end{enumerate}

We answer the first question affirmatively using quantitative case studies in the next subsection, where RMSDs ranges from 0.119\r{A} to 34.162\r{A}. The magnitude  of  the  correlation  between  our  robustness  measure  and  the  RMSD between the predicted structure and the ground truth is 0.917, that is, the predictions with low robustness measure cannot be trusted. The second question is also answered affirmatively in the next section where we compare results from BLOSUM62 sequence similarity measures of 20 and 40.
We answer the third question in the negative and thereby establish the need for our distance geometry based distance measures for adversarial attacks on protein folding networks. Finally, we show that the exact approach can be prohibitively expensive and can take more than 1 day for a single sequence on a 80GB A100 GPU. On the other hand, our approximation strategy takes about 2 hours and 10 minutes, and produces RMSD values that are no less than 0.69 times the exact RMSD values.

\subsection{Quantitative Case Studies}
We performed adversarial attacks on 10 naturally occurring proteins. The adversarial sequences and the original sequences differed in BLOSUM62 distance of at most 20, and are thus biologically close to each other.

\subsubsection{6NT5\_1: Stimulator of interferon protein in human beings}
As shown in Fig.~\ref{Fig:case_studies}(a), the alignment involving all residues has a RMSD score of 5.026\r{A} showing that the structures for the original input and its adversarial perturbation are quite different.
\subsubsection{SFP1\_YEAST: Yeast transcription factor}
Figure~\ref{Fig:intro}(left) shows its structure and the structure of its adversarial perturbation aligned together with a high RMSD score of 25.641\r{A}.
\subsubsection{6OEU\_1: Protein patched homolog 1 } Figure~\ref{Fig:case_studies}(b) shows the aligned structures for the sequence and its adversarial perturbation. The two sequences were aligned using PyMOL and achieved a RMSD of  34.162\r{A}.
\subsubsection{NP\_010457.3: Yeast translation termination factor GTPase eRF3}
Figure~\ref{Fig:case_studies}(c) shows the structure of the original sequence and the structure of the adversarial sequence  aligned together with a high RMSD of 6.870\r{A}.
\subsubsection{2BBJ\_1: Eubacteria CorA Mg2+ transporter}
Figure~\ref{Fig:case_studies}(d) shows the structures for the original and adversarial sequences with a high RMSD score of 6.76\r{A}.
\subsubsection{2ROP\_1: Human Copper-transporting ATPase 2}
As illustrated in Figure~\ref{Fig:case_studies}(e), the alignment between the structure corresponding to this sequence and its adversarial perturbation has a high RMSD score of 8.495.
\subsubsection{6NT3\_1: Sodium channel protein type 9}
As shown in Figure~\ref{Fig:case_studies}(f), the alignment between the structures corresponding to the original and the adversarial sequences has a high RMSD score of 15.361\r{A}.
\subsubsection{6O77\_1: Transient receptor potential catio channel subfamily M member 8}
The structures corresponding to the original  and the adversarial sequences are very different with a RMSD of 19.001\r{A}.

\subsubsection{2QK4\_1: Trifunctional purine biosynthetic protein adenosine-3}
The structures corresponding to the original and the adversarial sequences are quite similar and their alignment has a RMSD of only 0.730\r{A}.

\subsubsection{2QJF\_1: Bifunctional 3'-phosphoadenosine 5'-phosphosulfate synthetase 1}
Figure~\ref{Fig:intro}(right) shows the alignment between the structures corresponding to the original and the adversarial sequences with a really small RMSD score of 0.119\r{A}.

\subsection{Comparison with Ground Truth and DeepAccNet}
An inaccurate prediction has low robustness measure (inverse of RMSD) as shown in  Table~\ref{Table:DeepAcc}. We compare the RMSD distances associated with our robustness measure with the ground truth, and find a  correlation of 0.917.
We also compared the RMSD distance between the ground truth and the predicted structure with the local distance difference test (lDDT) prediction implemented in the DeepAccNet~\cite{hiranuma2021improved} tool. Our results are shown in Table~\ref{Table:DeepAcc}. 

\begin{table}[h!]
    \centering
    \begin{tabular}{ccc}
    \toprule
    Computed RMSD  & Ground Truth & DeepAccNet  \\
    (1/robustness-measure)  & RMSD  & Accuracy \\
    \midrule
      5.026   & 4.112  &  0.4631 \\
      25.641   & - &  0.10254 \\
      34.162   & 37.276  & 0.1755 \\
      6.870   & -  & 0.36 \\
      6.760   &  7.593 & 0.6704 \\
      8.495   &  12.355 &  0.4006 \\
      15.361   & 25.453  & 0.1847 \\
       19.001  & 10.933 &  0.3481 \\
       0.730  &  1.299 &  0.7324 \\
        0.119  &  0.922 &  0.6885 \\
    \bottomrule
    \end{tabular}
    \caption{High robustness measure (low RMSD) corresponds to low ground truth RMSD. The correlation measure between robustness measure and ground truth has magnitude 0.917 while the correlation measure between the DeepAccNet lDDT predictions and the ground truth has magnitude 0.869.}
    \label{Table:DeepAcc}
\end{table}

\subsection{Impact of BLOSUM Sequence Similarity Distances}
We have introduced the use of protein sequence similarity measures such as BLOSUM matrices to ensure that adversarial attacks create sequences that are biologically similar. A natural question to investigate is how a change in the bound on biological similarity changes the adversarial sequence. In order to study this, we launched adversarial attacks on the 10 proteins studied in the previous section using thresholds of 20 and 40 BLOSUM62 distances, and computed the RMSD of the two adversarial attacks. 
\begin{figure}[h!]
\centering 
\includegraphics[width=0.6 \columnwidth]{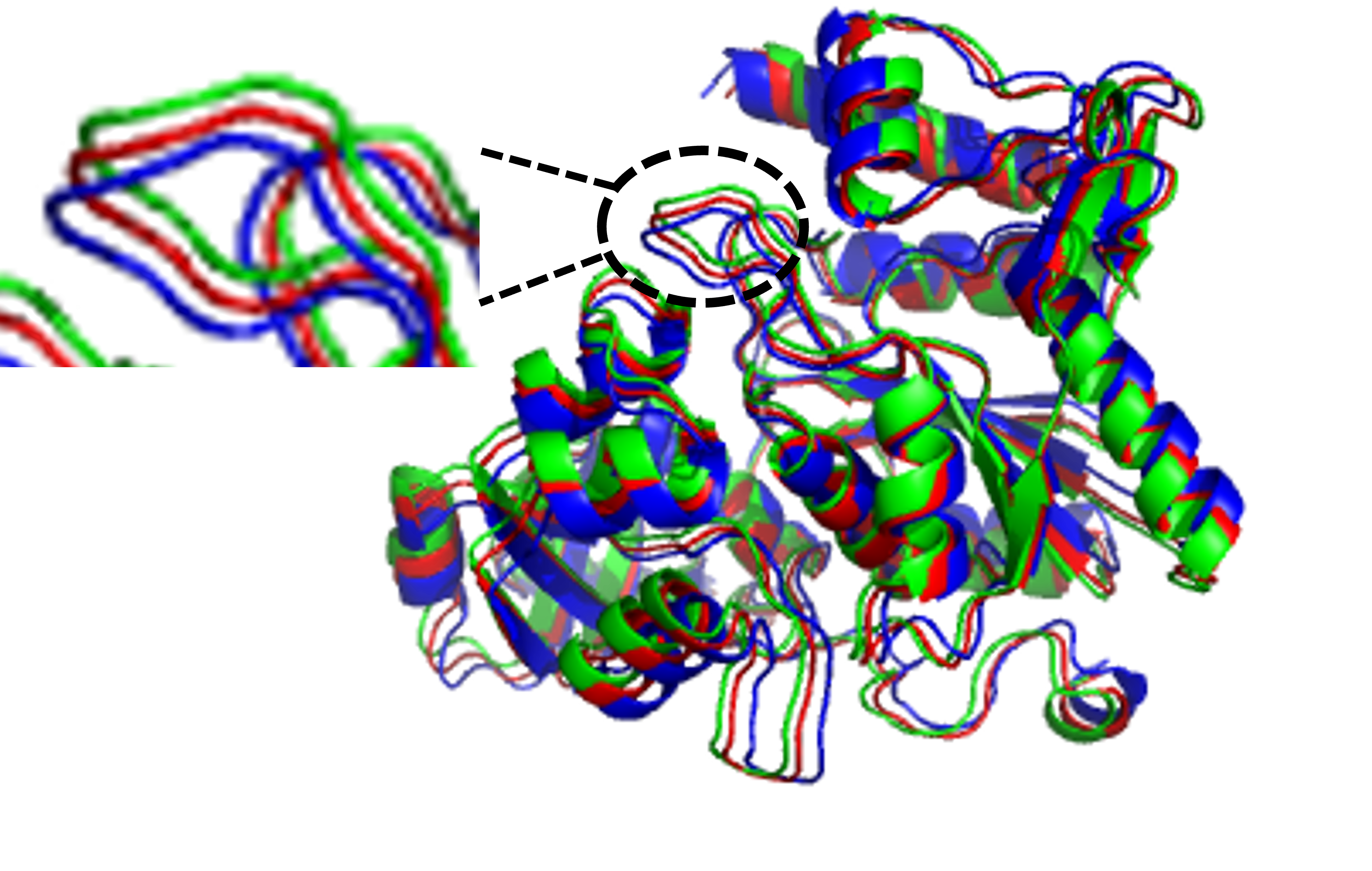} 
\caption{The adversarial sequence for 2QJF\_1 protein  with a BLOSUM62 threshold of 40 (blue) is farther away from the adversarial sequence with a BLOSUM62 threshold of 20 (red). See inset zoomed picture. The original sequence is shown in green.}
\label{Fig:proteinrmsd}
\end{figure}
\begin{table}[h]
    \centering
    \begin{tabular}{crr}
    \toprule
    Protein & \multicolumn{1}{c}{BLOSUM62-20} & \multicolumn{1}{c}{BLOSUM62-40}   \\
     ID & \multicolumn{1}{c}{RMSD (\r{A})} & \multicolumn{1}{c}{RMSD (\r{A})} \\
    \midrule 
      6NT5\_1   & 5.026 &  9.756 \\
      SFP\_1   &  25.641 &  26.355  \\
      6OEU\_1   & 34.162  & 35.295\\
      NP\_010457.3  &  6.870 &   9.609\\
      2BBJ\_1  & 6.760  & 7.504 \\
      2ROP\_1  &  8.495 &  10.969\\
      6NT3\_1  &    15.361  & 17.714  \\
      6O77\_1  & 19.001  & 26.532  \\
      2QK4\_1  &    0.730 &  0.893  \\
      2QJF\_1  & 0.119  &  0.266    \\
\bottomrule
    \end{tabular}
    \caption{Comparison of RMSD for adversarial examples using thresholds of 20 and 40 BLOSUM62 distances.}
    \label{table:blosum_vs_nonblosum}
\end{table}

As shown in Table~\ref{table:blosum_vs_nonblosum}, the RMSD for threshold of 20 BLOSUM62 distance is smaller than the RMSD for threshold of 40 BLOSUM62 distance.

Figure~\ref{Fig:proteinrmsd} shows the adversarial examples obtained from one protein sequence. In particular, we observe that the adversarial attack using a threshold of 40 for the BLOSUM62 distance makes more  changes to the structure than the adversarial attack constrained by a threshold of 20 for the BLOSUM62 distance.

\subsection{Co-ordinate  vs.  Distance Geometry}

We have introduced the idea of using  distance geometry instead of geometric co-ordinates for adversarial attacks on PFNNs. We present experimental results to show that the use of geometric co-ordinates leads to adversarial structures that are different from the original structure in geometric co-ordinates, but such structures can be translated, rotated and flipped to nearly the same structure. We illustrate this with an adversarial attack example in Fig.~\ref{Fig:geometric}.
\begin{figure}[h]
    \centering
    \includegraphics[width=0.8\columnwidth]{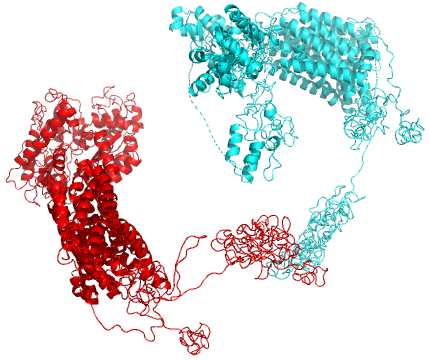}
    \caption{Adversarial attacks using geometric co-ordinates produce adversarial structures that are essentially translated, rotated or mirror images of the original structure.}
    \label{Fig:geometric}
\end{figure}

Table~\ref{table:geometric} shows the quantitative results of computing adversarial attacks using distance geometry and using 3-dimensional co-ordinate geometry, respectively. Geometric co-ordinates are not robust to rigid-body motion, and hence high loss functions during attacks do not correlate well with high RMSDs after alignment.

\begin{table}[h]
    \centering
    \begin{tabular}{crr}
    \toprule
    Protein & \multicolumn{1}{c}{Distance} & \multicolumn{1}{c}{Three-dimensional}   \\
     ID & \multicolumn{1}{c}{Geometry} & \multicolumn{1}{c}{Co-ordinates}    \\
      & \multicolumn{1}{c}{RMSD (\r{A})} & \multicolumn{1}{c}{RMSD (\r{A})} \\
    \midrule 
      6NT5\_1   & 5.026 &   3.685 \\
      SFP\_1   &  25.641 &   21.530  \\
      6OEU\_1   & 34.162  &  18.132 \\
      NP\_010457.3  &  6.870 &   0.165 \\
      2BBJ\_1  & 6.760  &  0.293 \\
      2ROP\_1  &  8.495 &  5.014 \\
      6NT3\_1  &    15.361  &  14.933  \\
      6O77\_1  & 19.001  &   17.558 \\
      2QK4\_1  &    0.730 &   0.042 \\
      2QJF\_1  & 0.119  &   0.104   \\
\bottomrule
    \end{tabular}
    \caption{Comparison of RMSD for distance geometry vs. co-ordinate geometry representations after alignment using PyMOL. Our distance geometry based approach produces adversarial examples that have up to 37 times larger RMSD than those produced using co-ordinate geometry. }
    \label{table:geometric}
\end{table}

\subsection{Exact vs. Approximation Distance Measures}

We have introduced an approximation strategy for computing adversarial examples for protein folding neural networks using distance geometry. Theoretically, this leads to a simplification from $\mathcal{O}(n^2)$ computations for computing the exact Euclidean distance matrix to $\mathcal{O}(n)$ computations for our approximation strategy.
Our approximate distance geometry approach is provably approximate and produces no less than 0.69 times the exact RMSD (Table~\ref{table:approx2}) while achieving 14X speedups (Table~\ref{table:approx1}) in our experiments on a 80GB A100 GPU system.

\begin{table}[h]
    \centering
    \begin{tabular}{crr}
    \toprule
    Protein & \multicolumn{1}{c}{Exact} & \multicolumn{1}{c}{Approximate}   \\
     ID & \multicolumn{1}{c}{Distance Matrices} & \multicolumn{1}{c}{Distance Matrices}    \\
      & \multicolumn{1}{c}{RMSD (\r{A})} & \multicolumn{1}{c}{RMSD (\r{A})} \\
    \midrule 
      6NT5\_1   &   9.053 &  5.026    \\
      SFP\_1   &    28.452 &  25.641    \\
      6OEU\_1   &  36.647 &   34.162    \\
      NP\_010457.3  &  11.529 &  6.870     \\
      2BBJ\_1  & 8.868   &   6.760   \\
      2ROP\_1  &  10.595 &  8.495   \\
      6NT3\_1  &  25.640  &   15.361   \\
      6O77\_1  &  21.813 & 19.001     \\
      2QK4\_1  &     1.096 &  0.730     \\
      2QJF\_1  &     0.146 & 0.119       \\
\bottomrule
    \end{tabular}
    \caption{Comparison of RMSD for exact and approximate Euclidean distance measures after alignment using PyMOL.}
    \label{table:approx2}
\end{table}

\begin{table}[h]
    \centering
    \begin{tabular}{ccccc}
    \toprule
    & \multicolumn{2}{c}{Exact} &  \multicolumn{2}{c}{Approximate} \\
    Protein & \multicolumn{1}{c}{Runtime} & \multicolumn{1}{c}{Memory} & \multicolumn{1}{c}{Runtime} & \multicolumn{1}{c}{Memory}  \\
     ID & \multicolumn{1}{c}{(min)} & \multicolumn{1}{c}{(MB)} & \multicolumn{1}{c}{(min)} & \multicolumn{1}{c}{(MB)}    \\
    \midrule 
      6NT5\_1   &  146.58 & 25,003 & 22.13 & 18,433   \\
      SFP\_1   &    372.23 &   52,809 & 43.27 & 29,229\\
      6OEU\_1   &   840.37 &  79,593 &  130.28 & 79,635 \\
      NP\_010457  & 421.08   & 58,021 & 33.03 & 34,265  \\
      2BBJ\_1  &  192.17 & 27,867  &  21.01 & 18,315 \\
      2ROP\_1  &    66.54 & 19,011  & 11.18 & 12,035\\
      6NT3\_1  &     1375.31  & 77,141  & 102.08 & 77,171  \\
      6O77\_1  &    1476.17 &  74,965  & 105.10  & 74,965 \\
      2QK4\_1  &    185.28  &  39,031 & 21.26 & 31,003  \\
      2QJF\_1  &    124.43 &  21,725 & 13.23 & 20,713    \\
\bottomrule
    \end{tabular}
    \caption{Comparison of exact vs. approximate approach in terms of runtime and memory.}
    \label{table:approx1}
\end{table}

\section{Conclusions \& Future Work}
The progress made in recent years on the prediction of protein folding structures promises to enable profound advances in the understanding of diseases, and the design of drugs and therapeutics. However, until these predictions are shown to be robust, the grand challenge of predictive protein folding persists. In this paper, we have presented the first step in this direction by demonstrating that PFNNs are susceptible to adversarial attacks. 
  RMSD in the  protein structure predicted by RoseTTAFold~\cite{baek2021accurate}  ranges from  0.119\r{A} to 34.162\r{A} when the adversarial perturbations are bounded by 20 units in the BLOSUM62 distance.
    We define a new robustness measure for the predicted structure of a protein sequence, and
    show that the magnitude of the correlation (0.917) between our robustness measure and the RMSD between the predicted structure and the ground truth is high. 
     Hence, our robustness measure can be used to responsibly deploy PFNNs.

\bibliography{aaai22}

\begin{thebibliography}{18}
\providecommand{\natexlab}[1]{#1}

\bibitem[{Andriushchenko et~al.(2020)Andriushchenko, Croce, Flammarion, and
  Hein}]{andriushchenko2020square}
Andriushchenko, M.; Croce, F.; Flammarion, N.; and Hein, M. 2020.
\newblock Square attack: a query-efficient black-box adversarial attack via
  random search.
\newblock In \emph{European Conference on Computer Vision}, 484--501. Springer.

\bibitem[{Baek et~al.(2021)Baek, DiMaio, Anishchenko, Dauparas, Ovchinnikov,
  Lee, Wang, Cong, Kinch, Schaeffer et~al.}]{baek2021accurate}
Baek, M.; DiMaio, F.; Anishchenko, I.; Dauparas, J.; Ovchinnikov, S.; Lee,
  G.~R.; Wang, J.; Cong, Q.; Kinch, L.~N.; Schaeffer, R.~D.; et~al. 2021.
\newblock Accurate prediction of protein structures and interactions using a
  three-track neural network.
\newblock \emph{Science}, 373(6557): 871--876.

\bibitem[{Bai et~al.(2020)Bai, Zeng, Jiang, Wang, Xia, and
  Guo}]{bai2020improving}
Bai, Y.; Zeng, Y.; Jiang, Y.; Wang, Y.; Xia, S.-T.; and Guo, W. 2020.
\newblock Improving query efficiency of black-box adversarial attack.
\newblock In \emph{Computer Vision--ECCV 2020: 16th European Conference,
  Glasgow, UK, August 23--28, 2020, Proceedings, Part XXV 16}, 101--116.
  Springer.

\bibitem[{Blumenthal(1970)}]{blumenthal1970theory}
Blumenthal, L. 1970.
\newblock \emph{Theory and Applications of Distance Geometry}.
\newblock AMS Chelsea Publishing Series. Chelsea Publishing Company.
\newblock ISBN 9780828402422.

\bibitem[{Chan et~al.(2019)Chan, Shan, Dahoun, Vogel, and
  Yuan}]{chan2019advancing}
Chan, H.~S.; Shan, H.; Dahoun, T.; Vogel, H.; and Yuan, S. 2019.
\newblock Advancing drug discovery via artificial intelligence.
\newblock \emph{Trends in pharmacological sciences}, 40(8): 592--604.

\bibitem[{Chothia and Lesk(1986)}]{chothia1986relation}
Chothia, C.; and Lesk, A.~M. 1986.
\newblock The relation between the divergence of sequence and structure in
  proteins.
\newblock \emph{The EMBO journal}, 5(4): 823--826.

\bibitem[{Cordes et~al.(2000)Cordes, Burton, Walsh, McKnight, and
  Sauer}]{Cordes_2000}
Cordes, M. H.~J.; Burton, R.~E.; Walsh, N.~P.; McKnight, C.~J.; and Sauer,
  R.~T. 2000.
\newblock An evolutionary bridge to a new protein fold.
\newblock \emph{Nature Structural Biology}, 7(12): 1129--1132.

\bibitem[{Croce et~al.(2020)Croce, Andriushchenko, Sehwag, Flammarion, Chiang,
  Mittal, and Hein}]{croce2020robustbench}
Croce, F.; Andriushchenko, M.; Sehwag, V.; Flammarion, N.; Chiang, M.; Mittal,
  P.; and Hein, M. 2020.
\newblock Robustbench: a standardized adversarial robustness benchmark.
\newblock \emph{arXiv preprint arXiv:2010.09670}.

\bibitem[{Croce and Hein(2021)}]{croce2021mind}
Croce, F.; and Hein, M. 2021.
\newblock Mind the box: $ l\_1 $-APGD for sparse adversarial attacks on image
  classifiers.
\newblock \emph{arXiv preprint arXiv:2103.01208}.

\bibitem[{Dayhoff, Schwartz, and Orcutt(1978)}]{dayhoff197822}
Dayhoff, M.; Schwartz, R.; and Orcutt, B. 1978.
\newblock 22 a model of evolutionary change in proteins.
\newblock \emph{Atlas of protein sequence and structure}, 5: 345--352.

\bibitem[{Goodfellow, McDaniel, and Papernot(2018)}]{goodfellow2018making}
Goodfellow, I.; McDaniel, P.; and Papernot, N. 2018.
\newblock Making machine learning robust against adversarial inputs.
\newblock \emph{Communications of the ACM}, 61(7): 56--66.

\bibitem[{Henikoff and Henikoff(1992)}]{henikoff1992amino}
Henikoff, S.; and Henikoff, J.~G. 1992.
\newblock Amino acid substitution matrices from protein blocks.
\newblock \emph{Proceedings of the National Academy of Sciences}, 89(22):
  10915--10919.

\bibitem[{Hiranuma et~al.(2021)Hiranuma, Park, Baek, Anishchenko, Dauparas, and
  Baker}]{hiranuma2021improved}
Hiranuma, N.; Park, H.; Baek, M.; Anishchenko, I.; Dauparas, J.; and Baker, D.
  2021.
\newblock Improved protein structure refinement guided by deep learning based
  accuracy estimation.
\newblock \emph{Nature communications}, 12(1): 1--11.

\bibitem[{Jones, Taylor, and Thornton(1992)}]{jones1992rapid}
Jones, D.~T.; Taylor, W.~R.; and Thornton, J.~M. 1992.
\newblock The rapid generation of mutation data matrices from protein
  sequences.
\newblock \emph{Bioinformatics}, 8(3): 275--282.

\bibitem[{Jumper et~al.(2021)Jumper, Evans, Pritzel, Green, Figurnov,
  Ronneberger, Tunyasuvunakool, Bates, {\v{Z}}{\'\i}dek, Potapenko
  et~al.}]{jumper2021highly}
Jumper, J.; Evans, R.; Pritzel, A.; Green, T.; Figurnov, M.; Ronneberger, O.;
  Tunyasuvunakool, K.; Bates, R.; {\v{Z}}{\'\i}dek, A.; Potapenko, A.; et~al.
  2021.
\newblock Highly accurate protein structure prediction with AlphaFold.
\newblock \emph{Nature}, 596(7873): 583--589.

\bibitem[{Rost(1999)}]{rost1999twilight}
Rost, B. 1999.
\newblock Twilight zone of protein sequence alignments.
\newblock \emph{Protein engineering}, 12(2): 85--94.

\bibitem[{Sander and Schneider(1991)}]{sander1991database}
Sander, C.; and Schneider, R. 1991.
\newblock Database of homology-derived protein structures and the structural
  meaning of sequence alignment.
\newblock \emph{Proteins: Structure, Function, and Bioinformatics}, 9(1):
  56--68.

\bibitem[{Tuinstra et~al.(2008)Tuinstra, Peterson, Kutlesa, Elgin, Kron, and
  Volkman}]{tuinstra2008interconversion}
Tuinstra, R.~L.; Peterson, F.~C.; Kutlesa, S.; Elgin, E.~S.; Kron, M.~A.; and
  Volkman, B.~F. 2008.
\newblock Interconversion between two unrelated protein folds in the
  lymphotactin native state.
\newblock \emph{Proceedings of the National Academy of Sciences}, 105(13):
  5057--5062.

\end{thebibliography}

\end{document}